\pdfoutput=1 

\documentclass[runningheads]{llncs}
\usepackage{graphicx}
\usepackage{amsmath}\allowdisplaybreaks
\usepackage{amssymb}

\usepackage{xcolor}
\usepackage[colorlinks]{hyperref}
\AtBeginDocument{\hypersetup{
    citecolor=blue,
	linkcolor=blue,
	filecolor=magenta,      
	urlcolor=blue,
}}

\usepackage{thm-restate}

\newcommand{\compilehidecomments}{false}
\newcommand{\Pa}{{\it Pa}}
\newcommand{\pa}{{\it pa}}

\newcommand{\bp}{\boldsymbol{p}}
\newcommand{\bq}{\boldsymbol{q}}
\newcommand{\br}{\boldsymbol{r}}

\usepackage{ifthen}

\begin{document}

\ifthenelse{ \equal{\compilehidecomments}{true} }{%
	\newcommand{\wei}[1]{}
	\newcommand{\shi}[1]{}
}{
	\newcommand{\wei}[1]{{\color{blue!50!black}  [\text{Wei:} #1]}}
	\newcommand{\shi}[1]{{\color{red!60!black} [\text{Shi:} #1]}}
}

\newcommand{\compilefullversion}{true}
\ifthenelse{\equal{\compilefullversion}{false}}{%
	\newcommand{\OnlyInFull}[1]{}
	\newcommand{\OnlyInShort}[1]{#1}
}{%
	\newcommand{\OnlyInFull}[1]{#1}%
	\newcommand{\OnlyInShort}[1]{}%
}%

\title{Causal Inference for Influence Propagation ---  
		Identifiability of the Independent Cascade Model}
\titlerunning{Causal Inference for Influence Propagation}
%
\author{Shi Feng\inst{1}\orcidID{0000-0001-5517-0419} 
	\and
Wei Chen\inst{2}
}
\authorrunning{Shi Feng and Wei Chen}
%
\institute{IIIS, Tsinghua University, Beijing, China \\
\email{fengs19@mails.tsinghua.edu.cn}\\
\and
Microsoft Research, Beijing, China\\
\email{weic@microsoft.com}
}
\maketitle              
\begin{abstract}
Independent cascade (IC) model is a widely used influence propagation model for social networks. 
In this paper, we incorporate the concept and techniques from causal inference to study the
	identifiability of parameters from observational data 
	in extended IC model with unobserved confounding factors, which models more realistic
	propagation scenarios but is rarely studied in influence propagation modeling before.
We provide the conditions for the identifiability or unidentifiability of parameters for several 
	special structures including the Markovian IC model, semi-Markovian IC model, and IC model with a global 
	unobserved variable.
Parameter identifiability is important for other tasks such as influence maximization under the diffusion networks
	with unobserved confounding factors.

\end{abstract}
\section{Introduction}

Extensive research has been conducted studying the information and influence propagation behavior in social networks, with numerous propagation models and optimization algorithms proposed (cf. \cite{kempe03,chen2013information}).
Social influence among individuals in a social network is intrinsically a causal behavior --- one's action or behavior causes the change of the behavior of his or her friends in
	the network.
Therefore, it is helpful to view influence propagation as a causal phenomenon and apply the tools in causal inference to this domain.

In causal inference, one key consideration is the confounding factors caused by unobserved variables that affect the observed behaviors of individuals in the network.
For example, we may observe that user A adopts a new product and a while later her friend B adopts the same new product.
This situation could be because A influences B and causes B's adoption, but it could also be caused by an unobserved factor (e.g. an unknown information source) that affects both
	A and B.
Confounding factors are important in understanding the propagation behavior in networks, but so far the vast majority of influence propagation research does not consider
	confounders in network propagation modeling.
In this paper, we intend to fill this gap by explicitly including unobserved confounders into the model, and we borrow the research methodology from causal inference to
	carry out our research.
	
Causal inference research has developed many tools and methodologies to deal with such unobserved confounders, and one important problem in causal inference is to study
	the identifiability of the causal model, that is, if we can identify the certain effect of an intervention, or identify causal model parameters, from the observational data.
In this paper, we introduce the concept of identifiability in causal inference research to influence propagation research and study whether the propagation models can be
	identified from observational data when there are unobserved factors in the causal propagation model.
We propose the extend the classical independent cascade (IC) model to include unobserved causal factors, and consider the parameter identifiability problem for several common
	causal graph structures.
Our main results are as follows.
First, for the Markovian IC model, in which each unobserved variable may affect only one observed node in the network,  we show that it is fully identifiable.
Second, for the semi-Markovian IC model, in which each unobserved variable may affect exactly two observed nodes in the network, we show that as long as a local graph structure
	exists in the network, then the model is not parameter identifiable.
For the special case of a chain graph where all observed nodes form a chain and every unobserved variable affect two neighbors on the chain, the above result implies that we need
	to know at least $n/2$ parameters to make the rest parameters identifiable, where $n$ is the number of observed nodes in the chain.
We then show a positive result that when we know $n$ parameters on the chain, the rest parameters are identifiable.
Third, for the global hidden factor model where we have an unobserved variable that affects all observed nodes in the graph, we provide reasonable sufficient conditions so that the parameters are identifiable.
	
Overall, we view that our work starts a new direction to integrate rich research results from network propagation modeling and causal inference so that we could view
	influence propagation from the lens of causal inference, and obtain more realistic modeling and algorithmic results in this area.
For example, from the causal inference lens, the classical influence maximization problem~\cite{kempe03} of finding a set of $k$ nodes to maximize the total influence spread is
	really a causal intervention problem of forcing an intervention on $k$ nodes for their adoptions, and trying to maximize the causal effect of this intervention.
Our study could give a new way of studying influence maximization that works under more realistic network scenarios encompassing unobserved confounders. \OnlyInShort{Due to the limitation of space, the complete proofs of some of the theorems are placed in the full version \cite{DBLP:journals/corr/abs-2107-04224} on arXiv, and only outlines are given in this version.}

%
%
%
%
%

\section{Related Work}


\noindent
{\bf Influence Propagation Modeling.\ }
As described in \cite{chen2013information}, the main two models used to describe influence propagation are the independent cascade model and the linear threshold model. Past researches on influence propagation mostly focused on influence maximization problems, such as \cite{kempe03,tang2015influence}. In these articles, they select seed nodes online, observe the propagation in the network, and optimize the number of activated nodes after propagation by selecting optimal seed nodes. Also, some works are studying the seed-node set minimization problem, such as \cite{goyal2013minimizing}. However, in our work, we mainly consider restoring the parameters in the independent cascade model by observing the network propagation. After obtaining the parameters in the network, we can then base on this to accomplish downstream tasks including influence maximization and seed-node set minimization.

\vspace{2mm}
\noindent
{\bf Causal Inference and Identifiability.\ }
For general semi-Markovian Bayesian causal graphs, \cite{huang2006identifiability} and \cite{shpitser2006identification} have given two different algorithms to determine whether a do effect is identifiable, and these two algorithms have both soundness and correctness. \cite{huang2012pearl} also proves that the ID algorithm and the repeating use of the do calculus are equivalent, so for semi-Markovian Bayesian causal graphs, the do calculus can be used to compute all identifiable do effects.

In addition, for a special type of causal model, the linear causal model, articles \cite{drton2011global} and \cite{foygel2012half} have given some necessary conditions and sufficient conditions on whether the parameters in the graph are identifiable with respect to the structure of the causal graph. However, the necessary and sufficient condition for parameter identifiability problem is not addressed and it remains an open question.
In this paper, we study another special causal model derived from the IC model.
Since the IC model can be viewed as a Bayesian causal model when the graph structure is a directed acyclic graph and it has some special properties, we try to give some necessary conditions and sufficient conditions for the parameters to be identifiable under some special graph structures.

\section{Model and Problem Definitions}
\label{sec:model}

Following the convention in causal inference literature (e.g. \cite{Pearl09}), 
	we use capital letters ($U,V,X,\ldots$) to represent variables or a set of variables, and their corresponding lower-case letters to represent their values.
For a directed graph, we use $U$'s and $V$'s to represent nodes since each node will also be treated as a random variable in causal inference.
For a node $V_i$, we use $N^+(V_i)$ and $N^-(V_i)$ to represent the set of its out-neighbors and in-neighbors, respectively.
When the graph is directed acyclic (DAG), we refer to a node's in-neighbors as its parents and denote the set as $\Pa(V_i) = N^-(V_i)$.
When we refer to the actual values of the parent nodes of $V_i$, we use $\pa(V_i)$.
For a positive integer $k$, we use $[k]$ to denote $\{1, 2, \ldots, k\}$.
We use boldface letters to represent vectors, such as $\br = (r_1, r_2, \ldots, r_n) = (r_i)_{i\in [n]}$.


The classical {\em independent cascade model} \cite{kempe03} of influence diffusion in a social network is modeled as follows.
The social network is modeled as a directed graph $G=(V,E)$, where $V=\{V_1,V_2,\cdots,V_n\}$ is the set of nodes representing individuals in the social network, and
	$E \subseteq V \times V$ is the set of directed edges representing the influence relationship between the individuals.
Each edge $(V_i, V_j) \in E$ is associated with an influence probability $p(i,j)\in (0,1]$ (we assume that $p(i,j)=0$ if $(V_i, V_j) \notin E$).
Each node is either in state $0$ or state $1$, representing the idle state and the active state, respectively. 
At time step $0$, a {\em seed set} $S_0 \subseteq V$ of nodes is selected and activated (i.e. their states are set to $1$), and all other nodes are in state $0$.
The propagation proceeds in discrete time steps $t=1, 2, \ldots$.
Let $S_t$ denote the set of nodes that are active by time $t$, and let $S_{-1}=\emptyset$.
At any time $t = 1, 2, \ldots$, the newly activated node $V_i \in S_{t-1} \setminus S_{t-2}$ tries to activate each of its inactive outgoing neighbors $V_j \in N^+(V_i)$, and
	the activation is successful with probability $p(i,j)$.
If successful, $V_j$ is activated at time $t$ and thus $V_j \in S_t$.
The activation trial of $V_i$ on its out-neighbor $V_j$ is independent of all other activation trials.
Once activated, nodes stay as active, that is, $S_{t-1} \subseteq S_t$.
The propagation process ends at a step when there are no new nodes activated.
It easy to see that the propagation ends in at most $n-1$ steps, so we use $S_{n-1}$ to denote the final set of active nodes after the propagation.

Influence propagation is naturally a result of causal effect --- one node's activation causes the activation of its outgoing neighbors.
If the graph is directed and acyclic, then the IC model on this graph can be equated to a Bayesian causal model.
In fact, we can consider each node in the IC model as a variable, and for a node $V_i$, it takes the value determined by $P(V_i=1|\pa(V_i))=1-\prod_{j:V_j\in \Pa(V_i),v_j=1\text{ in }\pa(V_i)}(1-p_{j,i})$. Obviously, this is equivalent to our definition in the IC model.
IC model is introduced in~\cite{kempe03} to model influence propagation in social networks, but in general, it can model the causal effects among binary random variables.
In this paper, we mainly consider the directed acyclic graph (DAG) setting, which is in line with the causal graph setting in the causal inference literature~\cite{Pearl09}.
\OnlyInFull{We will discuss the extension to general cyclic graphs or networks in the appendix.}\OnlyInShort{We discuss the extension to general cyclic graphs or networks in the full version \cite{DBLP:journals/corr/abs-2107-04224}.}

All variables $V_1, V_2, \ldots, V_n$ are observable, and we call them {\em observed variables}.
They correspond to observed behaviors of individuals in the social network.
There are also potentially many {\em unobserved (or hidden) variables} that affecting individuals' behaviors. 
We use $U=\{U_1, U_2, \ldots \}$ to represent the set of unobserved variables.
In the IC model, we assume each $U_i$ is a binary random variable with probability $r_i$ to be $1$ and probability $1-r_i$ to be $0$, and all unobserved variables are
mutually independent.
We allow unobserved variables $U_i$'s to have directed edges pointing to the observed variables $V_j$'s, but we do not consider directed edges among the
unobserved variables in this paper.
If $U_i$ has a directed edge pointing to $V_j$, we usually use $q_{i,j}$ to represent the parameter on this edge.
It has the same semantics as the $p_{i,j}$'s in the classical IC model: if $U_i=1$, then with probability $q_{i,j}$ $U_i$ successfully influence $V_j$ by setting its state to $1$, and with probability $1-q_{i,j}$ 
$V_j$'s state is not affected by $U_i$, and this influence or activation effect is independent from all other activation attempts on other edges.
Thus, overall, in a network with unobserved or hidden variables, we use $G=(U,V,E)$ to represent the corresponding causal graph, where $U$ is the set of unobserved variables,
$V$ is the set of observed variables, and $E \subseteq (V\times V) \cup (U \times V)$ is the set of directed edges.
We assume that $G$ is a DAG, and the state of every unobserved variable $U_i$ is sampled from $\{0,1\}$ with parameter $r_i$, while the state of
every observed variable $V_j$ is determined by the states of its parents and the parameters on the incoming edges of $V_j$ following the IC model semantics.
In the DAG $G$, we refer to an observable node $V_i$ as a {\em root} if it has no observable parents in the graph.
Every root $V_i$ has at least one unobserved parent.
We use vectors $\bp, \bq, \br$ to represent parameter vectors associated with edges among observed variables, 
	edges from unobserved to observed variables, and unobserved nodes, respectively.
We refer to the model $M = (G=(U,V,E), \bp,\bq,\br)$ as the {\em causal IC model}.
When the distinction is needed, we use capital letters $P,Q,R$ to represent the parameter names, and lower boldface letters $\bp, \bq, \br$ to represent
	the parameter values.

In this paper, we focus on the {\em parameter identifiability} problem following the causal inference literature.
In the context of the IC model, the states of nodes $V=\{V_1, V_2, \ldots, V_n\}$ are observable while the states of $U = \{U_1, U_2, \ldots\}$ are unobservable.
We define parameter identifiability as follows.

\begin{definition}[Parameter Identifiability]
\label{def:identifiability}
Given a graph $G=(U,V,E)$, 
	we say that a set of IC model parameters $\Theta \subseteq P\cup Q \cup R $ on $G$ is {\em identifiable} if 
	after fixing the values of parameters outside $\Theta$ and fixing the observed probability distributions $P(V'=v')$ for all $V'\subseteq V$ and all $v' \in \{0,1\}^{|V'|}$,
	the values of parameters in $\Theta$ are uniquely determined.
We say that {\em the graph $G$ is parameter identifiable} if $\Theta = P\cup Q \cup R $.
Accordingly, the algorithmic problem of parameter identifiability is to derive the unique values of parameters in $\Theta$ given graph $G=(U,V,E)$, 
	the values of parameters outside $\Theta$, and
	the observed probability distributions $P(V'=v')$ for all $V'\subseteq V$ and all $v' \in \{0,1\}^{|V'|}$.
Finally, if the algorithm only uses a polynomial number of observed probability values $P(V'=v')$'s and runs in
	polynomial time, where both polynomials are with respect to the graph size, we say that the parameters in $\Theta$ are {\em efficiently identifiable}.
\end{definition}

Note that when there are no unobserved variables (except the unique unobserved variables for each root of the graph), 
	the problem is mainly to derive the parameters $p_{i,j}$'s from all observed $P(V'=v')$'s. 
In this case, the parameter identifiability problem bears similarity with the well-studied network inference problem~\cite{Gomez-RodriguezLK10,myers2010convexity,gomez2011uncovering,DuSSY12,netrapalli2012finding,abrahaotrace,daneshmand2014estimating,DuSGZ13,DuLBS14,narasimhanlearnability,pougetinferring,he2016learning}. The network inference problem focuses on using observed cascade data to derive the network structure and propagation parameters, and it emphasizes
	on the sample complexity of inferring parameters.
Hence, when there are no unobserved variables in the model, we could use the network inference methods to help to solve the parameter identifiability problem.
However, in real social influence and network propagation, there are other hidden 
	factors that affect the propagation and the resulting distribution.
Such hidden factors are not addressed in the network inference literature.
In contrast, our study in this paper is focusing on addressing these hidden factors in network inference, and thus we borrow the ideas from causal inference to study the identifiability problem under the IC model.
	
In this paper, we study three types of unobserved variables that could commonly occur in network influence propagation. 
They correspond to three types of IC models with unobserved variables, as summarized below.

\vspace{2mm}
\noindent
{\bf Markovian IC Model.\ }
In the {\em Markovian IC model}, each observed variable $V_i$ is associated with a unique unobserved variable $U_i$, and there is a directed edge from $U_i$ to $V_i$.
This models the scenario where each individual in the social network has some latent and unknown factor that affects its observed behavior.
We use $q_i$ to denote the parameter on the edge $(U_i, V_i)$. 
Note that the effect of $U_i$ on the activation of $V_i$ is determined by probability $r_i \cdot q_i$, and thus we treat $r_i=1$ for all $i\in [n]$, and focus on
	identifying parameters $q_i$'s.
Thus the graph $G=(U,V,E)$ has parameters $\bq = (q_i)_{i\in [n]}$, and $\bp = (p_{i,j})_{(V_i,V_j)\in E}$.
Figure \ref{fig:markovian_model} shows an example of a Markovian IC model.
If some $q_i=0$, it means that the observed variable $V_i$ has no latent variable influencing it, and it only receives influence from other observed variables.

\begin{figure}[htbp]
\centering
\begin{minipage}[t]{0.4\textwidth}
\centering
\includegraphics[width=\linewidth]{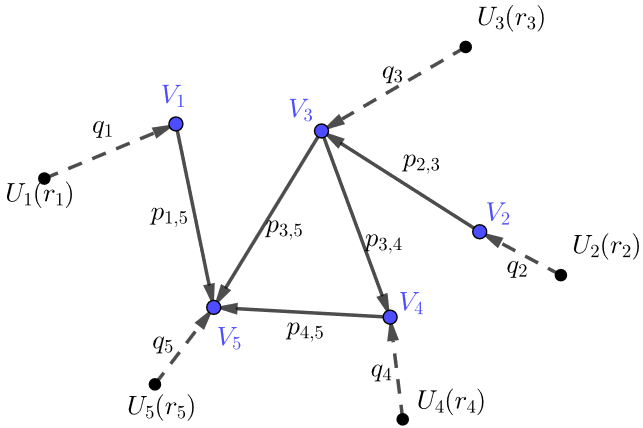}
\caption{A Markovian IC model with five nodes.}
\label{fig:markovian_model}
\end{minipage}
\begin{minipage}[t]{0.4\textwidth}
\centering
\includegraphics[width=\linewidth]{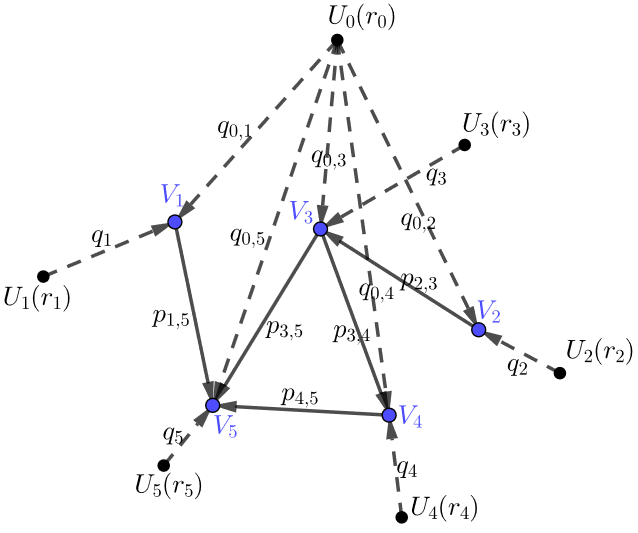}
\caption{A Markovian IC model with five nodes and a global unobserved variable.}
\label{fig:mixed_model}
\end{minipage}
\end{figure}


\vspace{2mm}
\noindent
{\bf Semi-Markovian IC Model.\ }
The second type of unobserved variables is the hidden variables connected to exactly two observed variables in the graph.
In particular, for every pair of nodes  $V_i,V_j \in V$, we allow one unobserved variable $U_{i,j}$ that has two edges, one pointing to $V_i$ and the other pointing to $V_j$.
This models the scenario that two individuals in the social network has a common unobserved confounder that may affect the behavior of two individuals.
We call this type of model {\em semi-Markovian IC model}, following the common terminology of the semi-Markovian model in the literature~\cite{Pearl09}.
In this model, each $U_{i,j}$ has a parameter $r_{i,j}$, and edges $(U_{i,j},V_i)$ and $(U_{i,j},V_j)$ have parameters $q_{i,j,1}$ and $q_{i,j,2}$ respectively.
Therefore, the graph has parameters $\br=(r_{i,j})_{(V_i,V_j)\in E}$, $\bq = (q_{i,j,1}, q_{i,j,2})_{(V_i,V_j)\in E}$, and $\bp = (p_{i,j})_{(V_i,V_j)\in E}$.


Within this model, we will pay special attention to a special type of graphs where the observed variables form a chain, i.e. $V_1\rightarrow V_2\rightarrow\cdots\rightarrow V_n$, and
	the unobserved variables always point to the two neighbors on the chain.
In this case, we use $U_i$ to denote the unobserved variable associated with edge $(V_i, V_{i+1})$, and the parameters on the edges $(U_i, V_i)$ and $(U_i, V_{i+1})$ are denoted
	as $q_{i,1}$ and $q_{i,2}$, respectively.
Figure~\ref{fig:semi_markovian_chain} represents this chain model.

\begin{figure}[htbp]
    \centering
    \includegraphics[width=0.7\linewidth]{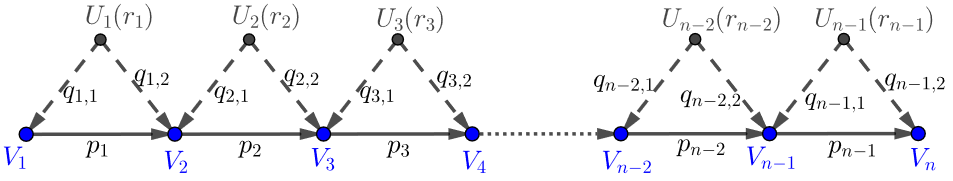}
    \caption{The semi-Markovian IC chain model.}
    \label{fig:semi_markovian_chain}
\end{figure}

\vspace{2mm}
\noindent
{\bf IC Model with A Global Unobserved Variable.\ }
The third type of hidden variables is a global unobserved variable $U_0$ that points to all observed variables in the network. 
This naturally models the global causal effect where some common factor affects all or most individuals in the network.
For every edge $(U_0,V_i)$, we use $q_{0,i}$ to represent its parameter.

Moreover, we can combine this model with the Markovian IC model, where we allow both unobserved variable $U_i$ for each individual and a global unobserved varoable $U_0$. 
Figure~\ref{fig:mixed_model} represents this model. 


\section{Parameter Identifiability of the Markovian IC Model}

For the Markovian IC model in which every observed variable has its own unobserved variable, we can fully identify the model parameters in most cases, as given
	by the following theorem.

\begin{theorem}[Identifiability of the Markovian IC Model]
For an arbitrary Markovian IC model $G=(U,V,E)$ with parameters $\bq = (q_i)_{i\in [n]}$ and $\bp = (p_{i,j})_{(V_i,V_j)\in E}$, 
	all the $q_i$ parameters are efficiently identifiable, and for every $i\in [n]$, if $q_i \ne 1$, then all $p_{j,i}$ parameters for $(V_j,V_i) \in E$ are 
	efficiently identifiable.
\end{theorem}
\begin{proof}
For an observed variable (node) $V_i$, suppose that its observed parents are $V_{i_1},V_{i_2},\cdots,V_{i_t}$. Therefore, we have
\begin{scriptsize}
\begin{align}
    &P(V_i=0|V_{i_1}=0,\cdots,V_{i_t}=0)=1-q_i, \label{eq:markovian_q}\\
    &P(V_i=0|V_{i_j}=1,V_{i_1}=0,\cdots,V_{i_{j-1}}=0,V_{i_{j+1}}=0,\cdots,V_{i_t}=0)=(1-q_i)(1-p_{i_j,i}).  \label{eq:markovian_p}
\end{align}
\end{scriptsize}
From Eq.\eqref{eq:markovian_q}, we can obtain the value of $q_i$.
Then if $q_i \ne 1$, from Eq.\eqref{eq:markovian_p}, we can derive the value of $p_{i_j,i}$.
Moreover, for each root node $V_i$, we can get $q_i$ by computing $q_i = P(V_i=1)$.
The computational efficiency is obvious.\qed
\end{proof}

The theorem essentially says that all parameters are identifiable under the Markovian IC model, except for the corner case where some $q_i=1$. 
In this case, the observed variable $V_i$ is fully determined by its unobserved parent $U_i$, so we cannot determine the influence from other observed parents of $V_i$ to $V_i$.
But the influence from the observed parents of $V_i$ to $V_i$ is not useful any way in this case, so the edges from the observed parents of $V_i$ to $V_i$ will not affect the
	causal inference in the graph and they can be removed.

\section{Parameter Identifiability of the Semi-Markovian IC Model}

Following the definition in the model section, we then consider the identifiability problem of the semi-Markovian models. We will demonstrate that in most cases, this model is not parameter identifiable. Actually, from \cite{shpitser2006identification} we know that the semi-Markovian Bayesian causal model is also not identifiable in general. Essentially, our conclusion is not related to their result. On the other side, we will show that with some parameters known in advance, the semi-Markovian IC chain model will be identifiable. 

\subsection{Condition on Unidentifiability of the Semi-Markovian IC Model}

More specifically, the following theorem shows the unidentifiability of the semi-Markovian IC model with a special structure in it. 

\begin{restatable}[Unidentifiability of the Semi-Markovian IC Model]{theorem}{thmunidsemi} \label{thm:unidentifiability}
Suppose in a general graph $G$, we can find the following structure. There are three observable nodes $V_1,V_2,V_3$ such that $(V_1,V_2)\in E, (V_2,V_3)\in E$ and unobservable $U_1,U_2$ with $(U_1,V_1),(U_1,V_2),(U_2,V_2),(U_2,V_3)\in E$. Suppose each of $U_1,U_2$ only has two edges associated to it, the three nodes $V_1,V_2,V_3$ can be written adjacently in a topological order of nodes in $U\cup V$. Then we can deduce that the graph $G$ is not parameter identifiable.
\end{restatable}
 Figure \ref{fig:two_triangle_ex} is an example of the structure described in the above theorem.

\begin{figure}[htbp]
	\centering
	\includegraphics[width=0.7\linewidth]{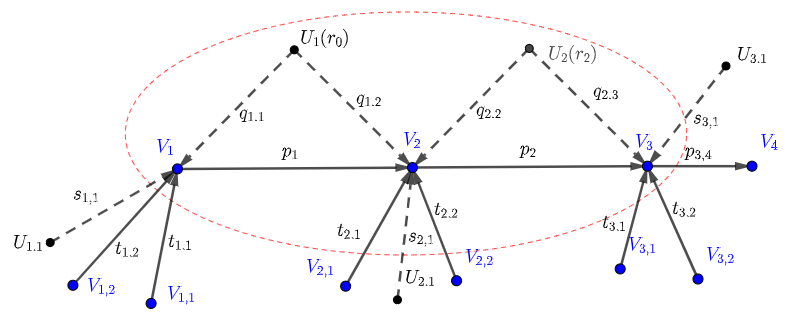}
	\caption{An example of the structure in Theorem \ref{thm:unidentifiability}.}
	\label{fig:two_triangle_ex}
\end{figure}

\begin{proof}[Outline]
To prove that the parameters in the model with this structure are not identifiable, we give two different sets of parameters directly. We show that these two different sets of parameters produce the same distribution of nodes in $V$, and thus the set of parameters is not identifiable by observing only the distribution of $V$. \OnlyInShort{The details of these two sets of parameters and the distributions they produce are included 
	in the full technical report~\cite{DBLP:journals/corr/abs-2107-04224}.}\OnlyInFull{The details of these two sets of parameters and the distributions they produce are included 
	in Appendix~\ref{app:thmunidsemi}.}
\qed
\end{proof}

\subsection{Identifiability of the Chain Model}

We now consider the chain model as described in Section~\ref{sec:model} and depicted in Figure \ref{fig:semi_markovian_chain}. In this structure, we present a conclusion of identifiability under the assumption that the valuations of some parameters are our prior knowledge.

We divide the parameters of the graph into four vectors 
\begin{align}
    &{\bq}_1=(q_{1,1},q_{2,1},\cdots,q_{n-1,1}),{\bq}_2=(q_{1,2},q_{2,2},\cdots,q_{n-1,2}),\\
    &{\bp}=(p_1,p_2,\cdots,p_{n-1}),{\br}=(r_1,r_2,\cdots,r_{n-1}).
\end{align}

For the chain model, our theorem below shows that once the parameters  $p_1$ is known, ${\bq}_2$ or ${\br}$ is known, the set consists of remaining parameters in the chain is efficiently identifiable.

\begin{restatable}[Identifiability of the Semi-Markovian IC Chain Model]{theorem}{thmidofsemi} \label{thm:id_of_semi}
Suppose that we have a semi-Markovian IC chain model with the graph $G=(U,V,E)$ and the IC parameters ${\bp}=(p_i)_{i\in [n-1]}$, ${\bq}_1=(q_{i,1})_{i\in[n-1]}$, 
${\bq}_2=(q_{i,2})_{i\in[n-1]}$ and ${\br}=(r_i)_{i\in[n-1]}$, and
suppose that all parameters are in the range $(0,1)$.
If the values of parameter $p_1$ is known, ${\bq}_2$ or ${\br}$ is known, then the remaining parameters are efficiently identifiable.
\end{restatable}

\begin{proof}[Outline]
We use induction to prove this theorem. Under the assumption that $p_1$ is known and ${\bq}_2$ or ${\br}$ is known, suppose $p_1,p_2,\cdots,p_{t-2}$, $r_1,r_2,\cdots,r_{t-2}$, $q_{1,1},q_{2,1},\cdots,q_{t-2,1}$ and $q_{1,2},q_{2,2},\cdots,q_{t-2,2},r_{ t-1}q_{t-1,1}$ has been determined by us, and we prove that $q_{t-1,1},r_{t-1},p_{t-1},q_{t-1,2}$ and $r_tq_{t,1}$ can also be determined. In fact, by the distribution of the first $t$ nodes on the chain we can obtain three different equations, and after substituting our known parameters, the inductive transition can be completed. It is worthy noting that this inductive process can also be used to compute the unknown parameters efficiently. 

The proof is lengthy because of the many corner cases considered and the need to discuss the cases $t=n,t=2$ and $2<t<n$. \OnlyInFull{The details of this proof are included 
	in Appendix~\ref{app:thmidofsemi}.} \hfill $\Box$
\end{proof}


According to Theorem \ref{thm:id_of_semi} we get that the semi-Markovian chain is parameter identifiable in the case that $n$ particular parameters are known. Simultaneously, by Theorem \ref{thm:unidentifiability}, we can show that if just less than $\lfloor\frac{n+1}{2}\rfloor$ parameters are known, then this semi-Markovian chain will not be parameter identifiable. Actually, if the chain model is parameter identifiable, utilizing Theorem \ref{thm:unidentifiability}, we know that for each $2\leq t\leq n-1$, at least one of parameters between $p_{t-1},p_t,r_{t-1},r_t,q_{t-1,1},q_{t-1,2},q_{t,1}$ and $q_{t,2}$ should be known. Therefore, we let $t=2,4,\cdots,2\lfloor\frac{n-1}{2}\rfloor$, we can deduce that at least $\lfloor\frac{n-1}{2}\rfloor$ should be known. Formally, we have the following collary of Theorem \ref{thm:unidentifiability} and Theorem \ref{thm:id_of_semi}.

\begin{corollary}
For a semi-Markovian IC chain model, if no more than $\lfloor\frac{n-1}{2}\rfloor$ parameters are known in advance, the remaining parameters are unidentifiable; if it is allowed to know $n$ parameters in advance, we can choose  $p_1,{\bq}_2$ or $p_1,{\br}$ to be known, then the remaining parameters are identifiable.
\end{corollary}

\section{Parameter Identifiability of Model with a Global Hidden Variable}

Next, we consider the case where there is a global hidden variable in the causal IC model, defined as those in Section \ref{sec:model}. If there is only one hidden variable $U_0$ in the whole model, we prove that the parameters in general in this model are identifiable; if there is not only $U_0$, the model is also Markovian, that is, there are also $n$ hidden variables $U_1,\cdots,U_n$ corresponding to $V_1,V_2,\cdots,V_n$, then the parameters in this model are identifiable if certain conditions are satisfied.

\subsection{Observable IC Model with Only a Global Hidden Variable}

Suppose the observed variables in the connected DAG graph $G=(U,V,E)$ are $V_1,V_2,\cdots,V_n$ in a topological order and there is a global hidden variable $U_0$ such that there exists an edge from $U_0$ to the node for each observable variable $V_i$. Suppose the activating probability of $U_0$ is $r$ and the activating probability from $U$ to $V_i$ is $q_i\in [0,1)$ (naturally, $q_1\neq 0$ and there are at least $3$ of nonzero $q_i$'s). Now we propose a theorem according to these settings.

\begin{restatable}[Identifiability of the IC Model with a Global Hidden Variable]{theorem}{thmidofglob} \label{thm:id_of_glob}
For an arbitrary IC model with a global hidden variable $G=(U,V,E)$ with parameters ${\bf q}=(q_i)_{i\in[n]}$, $r$ and ${\bf p}=(p_{i,j})_{(V_i,V_j)\in E}$ such that $q_i\ne 1,p_{i,j}\ne 1$ and $r\ne 1$ for $\forall i,j\in[n]$, all the parameters in ${\bf p},r$ and ${\bf q}$ are identifiable.
\end{restatable}


\begin{proof}[Outline]
We discuss this problem in two cases, the first one is the existence of two disconnected points $V_i,V_j,i<j$ in $V$ and $q_i,q_j\neq 0$. At this point we can use $1-q_j=\frac{P(V_1=0,V_2=0,\cdots,V_i=1,V_{i+1}=0,\cdots,V_j=0)}{P(V_1=0,V_2=0,\cdots,V_i=1,V_{i+1}=0,\cdots,V_{j-1}=0)}$ to solve out $q_j$, and then use $P(V_1=0,V_2=0,\cdots,V_j=0)$ and $P(V_1=0,V_2=0,\cdots,V_{j-1}=0)$ to solve out $r$.

After getting $r$, by the quotients of probabilities of propagating results, we can get all the parameters.

Another case is that there is no $V_i,V_j$ as described above. At this point there must exist three points $V_i,V_j,V_k$ that are connected with each other and $q_i,q_j,q_k\neq 0$. We observe the probabilities of different possible propagating results of these three points with all other nodes are $0$ after the propagation. From these, we can solve out $q_i,q_j,q_k$, and then solve out all parameters by the same method as in the first case.\qed
\end{proof}

\subsection{Markovian IC Model with a Global Hidden Variable (Mixed Model)}

Suppose the model is $G=(U,V,E)$, where $U=\{U_0,U_1,U_2,\cdots,U_n\}$, $V=\{V_1,V_2,\cdots,V_n\}$. Here, $V_1,V_2,\cdots,V_n$ are in a topological order. The parameters are $r_0$, ${\bf q}_0=(q_{0,i})_{i\in[n]}$, ${\bf q}=(q_i)_{i\in [n]}$ and ${\bf p}=(p_{i,j})_{(V_i,V_j)\in E}$.

\begin{restatable}[Identifiability of Markovian IC Model with a Global Hidden Variable (Mixed Model)]{theorem}{thmmixedmodel} \label{thm:mixedmodel}
For an arbitrary Markovian IC Model with a Global Hidden Variable $G=(U,V,E)$ with parameters $r_0$, ${\bf q}_0=(q_{0,i})_{i\in[n]}$, ${\bf q}=(q_i)_{i\in [n]}$ and ${\bf p}=(p_{i,j})_{(V_i,V_j)\in E}$, we suppose that all the parameters are not $1$. If $\exists i,j,k\in[n],i<j<k$ such that each pair in $V_i,V_j,V_k$ are disconnected and $q_{0,i},q_{0,j},q_{0,k}\neq 0$, then the parameters $q_{0,t},q_t$ and $p_{t,l}, l>t>k$ are identifiable. Moreover, if $V_i,V_j,V_k$ can be adjacently continuous in some topological order, i.e. $j=i+1,k=i+2$ without loss of generality, all the parameters are identifiable. 
\end{restatable}

\begin{proof}[Outline]
	Assuming that there exist $V_i,V_j,V_k$ that satisfy the requirements of the theorem, then we can write expressions for the distribution of these three parameters when all other nodes with subscripts not greater than $l$ are equal to $0$. In fact, we can see that with these $8$ expressions, we can solve for $P(V_1=0,\cdots,V_l=0,U_0=1)$ and $P(V_1=0,\cdots,V_l=0,U_0=0)$.

    Since we have $P(V_1=0,\cdots,V_l=0,U_0=1)=r\prod_{t=1}^l(1-q_t)(1-q_{0,t})$ and $P(V_1=0,\cdots,V_l=0,U_0=0)=(1-r)\prod_{t=1}^l(1-q_t)$,
    we will be able to obtain all the parameters very easily by dividing these equations two by two. This proof has some trivial discussion to show that this computational method does not fail due to corner cases. \qed
\end{proof}

Notice that the parameters in this model are identifiable when and only when a special three-node structure appears in it. Intuitively, this is because through this structure we can more easily obtain some information about the parameters, which does not contradict the intuition of Theorem \ref{thm:unidentifiability}.

\section{Conclusion}

In this paper, we study the parameter identifiability of the independent cascade model in influence propagation and show
	conditions on identifiability or unidentifiability for several classes of causal IC model structure.
We believe that the incorporation of observed confounding factors and causal inference techniques is important in the next
	step of influence propagation research and identifiability of the IC model is our first step towards this goal.
There are many open problems and directions in combining causal inference and propagation research.
For example, seed selection and influence maximization correspond to the intervention (or do effect) in causal inference,
	and how to compute such intervention effect under the network with unobserved confounders and how to do influence maximization is a very interesting research question.
In terms of identifiability, one can also investigate the identifiability of the intervention effect, or whether given some
	intervention effect one can identify more of such effects.
\OnlyInFull{One can also look into identifiability in the general cyclic IC models, for which we provide some initial discussions
	in Appendix~\ref{app:cyclic}, but more investigations are needed.}

%
%
%
\bibliographystyle{splncs04}
\bibliography{references}

\OnlyInFull{ \clearpage

\appendix

\section*{Appendix}

\section{Proof for Unidentifiability of the Semi-Markovian IC Model  (Theorem~\ref{thm:unidentifiability})}
\label{app:thmunidsemi}

\thmunidsemi*

\begin{proof}
	To prove that the parameters are unidentifiable, we will construct two different sets of valuations of parameters such that two distributions of values taken by the nodes in $V$ are the same. In fact, we assume that the parent nodes of $V_i$ are $V_{i,1},V_{i,2},\cdots,V_{i,t_i}$ and $U_{i,1},\cdots,U_{i,s_i}$ for $i=1,2,3$. And, the parameters are set as shown in Figure \ref{fig:two_triangle_ex}. 
	
	One set of parameters is that all the parameters are set to $0.5$. Another set of parameters is that all the parameters are set to be $0.5$ except $r_1,r_2,q_{1,1},q_{1,2},q_{2,2}$ and $q_{2,3}$. The exceptions are set to be
	\begin{align}
	&r_1=\frac{10r_2-7}{12r_2-10},q_{1,1}=\frac{1}{4r_1},\\
	&q_{1,2}=\frac{6r_2-5}{8r_2-8},q_{2,1}=\frac{1}{3-2r_2},q_{2,2}=\frac{1}{4r_2}
	\end{align}
	where $r_2$ is an arbitrary number between $\frac{1}{4}$ and $\frac{9}{14}$. Then we only need to prove that for an arbitrary distribution of the ancestors of $V_1,V_2,V_3$ except $U_1,U_2$, the distributions of $V_1,V_2,V_3$ are the same for the two parameter settings. This is because if this condition is satisfied, then the children of $V_1,V_2,V_3$ will not be affected by the differences of the parameter valuations because the parameters determining them are the same (only the distribution of $U_1,U_2$ will change but each of them only has two children that are in $\{V_1,V_2,V_3\}$). In fact, we have
	\begin{small}
		\begin{align}
		&P(V_1=1,V_2=1,V_3=1|pa(V_1),pa(V_2),pa(V_3))\\&=\frac{13+3P(V_2=1|pa(V_2)+3P(V_1=1|pa(V_2)(11+5P(V_2=1|pa(V_2))}{128},\\
		&P(V_1=1,V_2=1,V_3=0|pa(V_1),pa(V_2),pa(V_3))=\frac{23+9P(V_2=1|pa(V_2)}{64},\\
		&P(V_1=1,V_2=0,V_3=1|pa(V_1),pa(V_2),pa(V_3))\\&=\frac{3(1+5P(V_1=1|pa(V_1)))(P(V_2=1|pa(V_2))-1)}{128},\\
		&P(V_1=0,V_2=1,V_3=1|pa(V_1),pa(V_2),pa(V_3))\\&=\frac{3(1-P(V_1=1|pa(V_1)))(5P(V_2=1|pa(V_2))+3)}{64},\\
		&P(V_1=1,V_2=0,V_3=0|pa(V_1),pa(V_2),pa(V_3))\\&=\frac{3(1+5P(V_1=1|pa(V_1)))(1-P(V_2=1|pa(V_2)))}{128},\\
		&P(V_1=0,V_2=1,V_3=0|pa(V_1),pa(V_2),pa(V_3))\\&=\frac{3(1-P(V_1=1|pa(V_1)))(3+5P(V_2=1|pa(V_2)))}{64},\\
		&P(V_1=0,V_2=1,V_3=1|pa(V_1),pa(V_2),pa(V_3))\\&=\frac{3(1+5P(V_1=1|pa(V_1)))(1-P(V_2=1|pa(V_2)))}{128},\\
		&P(V_1=0,V_2=0,V_3=1|pa(V_1),pa(V_2),pa(V_3))\\&=\frac{15(1-P(V_1=1|pa(V_1)))(1-P(V_2=1|pa(V_2)))}{64},\\
		&P(V_1=0,V_2=0,V_3=1|pa(V_1),pa(V_2),pa(V_3))\\&=\frac{15(1-P(V_1=1|pa(V_1)))(1-P(V_2=1|pa(V_2)))}{64}.
		\end{align}
	\end{small}
	Notice that these equations are only related to $P(V_1=1|pa(V_1))$ and $P(V_2=1|pa(V_2))$, so until now, we have proved that the two sets of parameters produce two same distributions for observed variables and therefore $G$ is parameter unidentifiable at this point.\qed
\end{proof}

\section{Proof for the Chain Structure (Theorem~\ref{thm:id_of_semi})}
\label{app:thmidofsemi}
\thmidofsemi*

\begin{proof}
	
	In the analysis, we use the following notations.
	For observed nodes $V_1,V_2,$ $\ldots, V_t$, we use $\overline{V_1\cdots V_t}$ to represent their collective states as a bit string.
	For a bit string $\gamma$ of length $t$, we write 
	\begin{equation}
	\begin{aligned}
	&a^\gamma=P(\overline{V_1\cdots V_t}=\gamma), b^\gamma=P(\overline{V_1\cdots V_t}=\gamma,U_t=0),\\& c^\gamma=P(\overline{V_1\cdots V_t}=\gamma, U_t=1).
	\end{aligned}
	\end{equation}
	Note that $a^\gamma$ is observable, but $b^\gamma$ and $c^\gamma$ are not observable.
	
	We will use induction method to prove this result. More specifically, we want to prove the base step that $p_1,r_1,q_{1,1},q_{1,2}$ and $q_{2,1}r_2$ are known. Moreover, we will prove the induction step that if $p_1,p_2,\cdots,p_{t-2}$, $r_1,r_2,\cdots,r_{t-2}$, $q_{1,1},q_{2,1},\cdots,q_{t-2,1}$ and $q_{1,2},q_{2,2},\cdots,q_{t-2,2},r_{t-1}q_{t-1,1}$ are already known, we can compute $q_{t-1,1},r_{t-1},p_{t-1},q_{t-1,2}$ and $r_tq_{t,1}$ using distributions of observed variables.
	
	Initially, in order to complete the induction step, we divide the problem into two cases which are $t=n$ and $3\leq t\leq n-1$.
	
	Firstly, if $t=n$, we only need to compute $q_{t-1,1},r_{t-1},p_{t-1},q_{t-1,2}$. Actually, we have the following relations of $a^{\overline{\gamma 0}}$ and other known variables. If $\gamma[-1]=0$, we have
	\begin{equation}
	\begin{aligned}
	a^{\overline{\gamma 0}}&=b^\gamma+(1-q_{t-1,2})c^\gamma\\
	&=\frac{1-r_{t-1}}{1-q_{t-1,1}r_{t-1}}a^\gamma
	+\frac{r_{t-1}(1-q_{t-1,1})}{1-q_{t-1,1}r_{t-1}}(1-q_{t-1,2})a^\gamma\\
	&=a^\gamma-\frac{q_{t-1,2}r_{t-1}(1-q_{t-1,1})}{1-q_{t-1,1}r_{t-1}}a^\gamma
	\end{aligned}
	\label{equation:equation_of_gamma0}
	\end{equation}

	If $\gamma[-1]=1$ and $\gamma=\overline{\beta 1},\beta[-1]=0$, we have
	\begin{equation}
	\begin{aligned}
	a^{\overline{\gamma 0}}&=(1-p_{t-1})(a^\gamma-q_{t-1,2}c^\gamma)\\
	&=(1-p_{t-1})(a^\gamma-q_{t-1,2}r_{t-1}(a^\beta-(1-q_{t-1,1})(a^\beta-q_{t-2,2}c^\beta)))
	\end{aligned}
	\end{equation}
	
	If $\gamma[-1]=1$ and $\beta[-1]=1$, we have
	\begin{equation}
	\begin{aligned}
	a^{\overline{\gamma 0}}&=(1-p_{t-1})(a^\gamma-q_{t-1,2}c^\gamma)\\
	&=(1-p_{t-1})\\&(a^\gamma-q_{t-1,2}r_{t-1}(a^\beta-(1-p_{t-2})(1-q_{t-1,1})(a^\beta-q_{t-2,2}c^\beta)))
	\end{aligned}
	\end{equation}
	
	We should notice that $b^\beta,c^\beta$ are both determined by the parameters we already know, hence, we can see the three types of recursions as functions of $q_{t-1,2},q_{t-1,2}r_{t-1}$ and $p_{t-1}$. We only need to prove that we can solve the three parameters from these functions. In order to simplify the equations, we define some notations as below where $|\beta|=t-2$.
	\begin{equation}
	\begin{aligned}
	&\text{Co}_1(\beta)=\begin{cases}
	-a^{\beta}+(1-p_{t-2})(a^{\beta}-q_{t-2,2}c^{\beta})&\beta[-1]=1\\
	-a^{\beta}+(a^{\beta}-q_{t-2,2}c^{\beta})&\beta[-1]=0
	\end{cases}\\
	&\text{Co}_2(\beta)=\begin{cases}
	-r_{t-1}q_{t-1,1}(1-p_{t-2})(a^{\beta}-q_{t-2,2}c^{\beta})&\beta[-1]=1\\
	-r_{t-1}q_{t-1,1}(a^{\beta}-q_{t-2,2}c^{\beta})&\beta[-1]=0
	\end{cases}
	\end{aligned}
	\label{equation:def_of_Co}
	\end{equation}
	It is easy to verify that $\text{Co}_1(\beta)+\text{Co}_2(\beta)=-a^{\overline{\beta 1}}$. According to these, we have
	\begin{equation}
	\begin{aligned}
	\frac{a^{\overline{\beta_110}}}{a^{\overline{\beta_210}}}&=\frac{\text{Co}_1(\beta_1)q_{t-1,2}r_{t-1}+\text{Co}_2(\beta_1)q_{t-1,2}-\text{Co}_1(\beta_1)-\text{Co}_2(\beta_1)}{\text{Co}_1(\beta_2)q_{t-1,2}r_{t-1}+\text{Co}_2(\beta_2)q_{t-1,2}-\text{Co}_1(\beta_2)-\text{Co}_2(\beta_2)}
	\end{aligned}
	\end{equation}
	which is equivalent to
	\begin{equation}
	\begin{aligned}
	&(\frac{\text{Co}_1(\beta_1)}{a^{\overline{\beta_110}}}-\frac{\text{Co}_1(\beta_2)}{a^{\overline{\beta_210}}})q_{t-1,2}r_{t-1}
	+(\frac{\text{Co}_2(\beta_1)}{a^{\overline{\beta_110}}}-\frac{\text{Co}_2(\beta_2)}{a^{\overline{\beta_210}}})q_{t-1,2}\\
	&-(\frac{\text{Co}_1(\beta_1)+\text{Co}_2(\beta_1)}{a^{\overline{\beta_110}}}-\frac{\text{Co}_1(\beta_2)+\text{Co}_2(\beta_2)}{a^{\overline{\beta_210}}})=0
	\end{aligned}
	\label{equation:equation_of_gamma1}
	\end{equation}
	To show the relation of the coefficients of $q_{t-1,2}r_{t-1}$ and $q_{t-1,2}$, we prove two lemmas at first.

	\begin{lemma}
		For two arbitrary $0-1$ string $\beta_1,\beta_2$ with the same lengths, we prove that the equation $\frac{a^{\overline{\beta_2 0}}}{c^{\overline{\beta_2 0}}}=\frac{a^{\overline{\beta_1 1}}}{c^{\overline{\beta_1 1}}}$ is impossible. 
	\end{lemma}
	
	\begin{proof}
		Actually, if $\frac{a^{\overline{\beta_2 0}}}{c^{\overline{\beta_2 0}}}=\frac{a^{\overline{\beta_1 1}}}{c^{\overline{\beta_1 1}}}$, 
		we have $\frac{b^{\overline{\beta_1 1}}}{c^{\overline{\beta_1 1}}}=\frac{b^{\overline{\beta_2 0}}}{c^{\overline{\beta_2 0}}}$. 
		We can verify that 
		
		\begin{equation}
		\begin{aligned}
		\frac{b^{\overline{\beta_2 0}}}{c^{\overline{\beta_2 0}}} &= \frac{1-r_{t-1}}{r_{t-1}(1-q_{t-1,1})}\\
		&=\frac{b^{\overline{\beta_1 1}}}{c^{\overline{\beta_1 1}}}\\
		&=\begin{cases}
		\frac{1-r_{t-1}}{r_{t-1}}\frac{b^{\beta_1}+c^{\beta_1}-((1-q_{t-2,2})c^{\beta_1}+b^{\beta_1})}{b^{\beta_1}+c^{\beta_1}-(1-q_{t-1,1})(b^{\beta_1}+(1-q_{t-2,2})c^{\beta_1})}&\beta_1[-1]=0\\
		\frac{1-r_{t-1}}{r_{t-1}}\frac{b^{\beta_1}+c^{\beta_1}-(1-p_{t-2})((1-q_{t-2,2})c^{\beta_1}+b^{\beta_1})}{b^{\beta_1}+c^{\beta_1}-(1-p_{t-2})(1-q_{t-1,1})(b^{\beta_1}+(1-q_{t-2,2})c^{\beta_1})}&\beta_1[-1]=1
		\end{cases}
		\end{aligned}
		\end{equation}
		which is impossible if all the terms are not zero.\qed
	\end{proof}
	
	\begin{lemma}
		As we defined above, we have the following equation if the denominators are not zero:
		
		\begin{equation}
		\begin{aligned}
		&\frac{\frac{\text{Co}_1(\beta_1)}{a^{\overline{\beta_110}}}-\frac{\text{Co}_1(\beta_2)}{a^{\overline{\beta_210}}}}{\frac{\text{Co}_2(\beta_1)}{a^{\overline{\beta_110}}}-\frac{\text{Co}_2(\beta_2)}{a^{\overline{\beta_210}}}}=-\frac{q_{t-1,2}-1}{q_{t-1,2}r_{t-1}-1}
		\end{aligned}
		\end{equation}
		\label{equation:coef_lemma}
	\end{lemma}
	
	\begin{proof}
		Actually, we can compute the left hand side of Equation \ref{equation:coef_lemma} as below:
		
		\begin{equation}
		\begin{aligned}
		&\frac{\frac{\text{Co}_1(\beta_1)}{a^{\overline{\beta_110}}}-\frac{\text{Co}_1(\beta_2)}{a^{\overline{\beta_210}}}}{\frac{\text{Co}_2(\beta_1)}{a^{\overline{\beta_110}}}-\frac{\text{Co}_2(\beta_2)}{a^{\overline{\beta_210}}}}\\
		&=\frac{\frac{\text{Co}_1(\beta_1)}{\text{Co}_1(\beta_1)(q_{t-1,2}r_{t-1}-1)+\text{Co}_2(\beta_1)(q_{t-1,2}-1)}- \frac{\text{Co}_1(\beta_2)}{\text{Co}_1(\beta_2)(q_{t-1,2}r_{t-1}-1)+\text{Co}_2(\beta_2)(q_{t-1,2}-1)}}
		{\frac{\text{Co}_2(\beta_1)}{\text{Co}_1(\beta_1)(q_{t-1,2}r_{t-1}-1)+\text{Co}_2(\beta_1)(q_{t-1,2}-1)}-\frac{\text{Co}_2(\beta_2)}{\text{Co}_1(\beta_2)(q_{t-1,2}r_{t-1}-1)+\text{Co}_2(\beta_2)(q_{t-1,2}-1)}}\\
		&=\frac{(\text{Co}_1(\beta_1)\text{Co}_2(\beta_2)-\text{Co}_1(\beta_2)\text{Co}_2(\beta_1))(q_{t-1,2}-1)}{(\text{Co}_1(\beta_2)\text{Co}_2(\beta_1)-\text{Co}_1(\beta_1)\text{Co}_2(\beta_2))(q_{t-1,2}r_{t-1}-1)}\\
		&=-\frac{q_{t-1,2}-1}{q_{t-1,2}r_{t-1}-1}
		\end{aligned}
		\end{equation}
		
		So the lemma is proved. We should also notice that this lemma holds not only for $t=n$, it is true for $t=2,3,\cdots,n-1$ because $\text{Co}_1$ and $\text{Co}_2$ are defined the same as Equation \ref{equation:def_of_Co}.\qed
	\end{proof}

	Now we prove that Equation \ref{equation:equation_of_gamma0} and Equation \ref{equation:equation_of_gamma1} are two linear independent equations for unknown variables $q_{t-1,2}r_{t-1}$ and $q_{t-1,2}$. We only need to prove that the ratio of coefficients of $q_{t-1,2}r_{t-1}$ and $q_{t-1,2}$ are different in these two equations. Because otherwise, we have
	
	\begin{equation}
	\begin{aligned}
	-\frac{q_{t-1,2}-1}{q_{t-1,2}r_{t-1}-1}=-\frac{1}{q_{t-1,1}r_{t-1}}
	\end{aligned}
	\end{equation}
	
	which is impossible because $q_{t-1,1}r_{t-1}(q_{t-1,2}-1)-(q_{t-1,2}r_{t-1}-1)=(1-q_{t-1,2}r_{t-1})(1-q_{t-1,1}r_{t-1})+(1-r_{t-1})q_{t-1,1}q_{t-1,2}r_{t-1}>0$. Until now, we have proved that when $t$ is equal to $n$, we can compute $q_{t-1,1},r_{t-1},p_{t-1},q_{t-1,2}$ if the other parameters are already known.
	
	Now we consider the case that $3\leq t\leq n-1$, and we need to prove that if $t<n$, we can use $p_1,p_2,\cdots,p_{t-2}$, $r_1,r_2,\cdots,r_{t-2}$, $q_{1,1},q_{2,1},\cdots,q_{t-2,1}$ and $q_{1,2},q_{2,2},$ $\cdots,q_{t-2,2},r_{t-1}q_{t-1,1}$ to compute $q_{t-1,1},r_{t-1},p_{t-1},q_{t-1,2}$ and $r_tq_{t,1}$. First, we propose a lemma to illustrate the recurrence relationship of the distributions. 
	
	\begin{lemma}
		Suppose $\gamma$ is a $0-1$ string with $t-1$ items, then we have the following relations:
		
		\begin{equation}
		\begin{aligned}
		b^{\overline{\gamma 1}}=\begin{cases}
		(1-r_t)(a^\gamma-(1-p_{t-1})(b^\gamma+(1-q_{t-1,2})c^\gamma))&\gamma[-1]=1\\
		(1-r_t)(a^\gamma-(b^\gamma+(1-q_{t-1,2})c^\gamma))&\gamma[-1]=0
		\end{cases}
		\end{aligned}
		\end{equation}
		
		\begin{equation}
		\begin{aligned}
		b^{\overline{\gamma 0}}=\begin{cases}
		(1-r_t)(1-p_{t-1})(b^\gamma+(1-q_{t-1,2})c^\gamma)&\gamma[-1]=1\\
		(1-r_t)(b^\gamma+(1-q_{t-1,2})c^\gamma)&\gamma[-1]=0
		\end{cases}
		\end{aligned}
		\end{equation}
		
		\begin{equation}
		\begin{aligned}
		c^{\overline{\gamma 1}}=\begin{cases}
		r_t(a^\gamma-(1-p_{t-1})(1-q_{t,1})(b^\gamma+(1-q_{t-1,2})c^\gamma))&\gamma[-1]=1\\
		r_t(a^\gamma-(1-q_{t,1})(b^\gamma+(1-q_{t-1,2})c^\gamma))&\gamma[-1]=0
		\end{cases}
		\end{aligned}
		\end{equation}
		
		\begin{equation}
		\begin{aligned}
		c^{\overline{\gamma 0}}=\begin{cases}
		r_t(1-p_{t-1})(1-q_{t,1})(b^\gamma+(1-q_{t-1,2})c^\gamma)&\gamma[-1]=1\\
		r_t(1-q_{t,1})(b^\gamma+(1-q_{t-1,2})c^\gamma)&\gamma[-1]=0
		\end{cases}
		\end{aligned}
		\end{equation}
	\end{lemma}
	
	Moreover, we have the following relations of $a^{\overline{\gamma 0}}$ and other known variables. If $\gamma[-1]=0$, we have
	\begin{equation}
	\begin{aligned}
	a^{\overline{\gamma 0}}
	&=(1-r_tq_{t,1})(a^\gamma-\frac{q_{t-1,2}r_{t-1}(1-q_{t-1,1})}{1-q_{t-1,1}r_{t-1}}a^\gamma)
	\end{aligned}
	\label{equation:equation1_for_t_less_n}
	\end{equation}

	If $\gamma[-1]=1$ and $\gamma=\overline{\beta 1},\beta[-1]=0$, we have
	\begin{equation}
	\begin{aligned}
	a^{\overline{\gamma 0}}&=(1-r_tq_{t,1}) (1-p_{t-1})(a^\gamma-q_{t-1,2}c^\gamma)\\
	&=(1-r_tq_{t,1})(1-p_{t-1})\\&(a^\gamma-q_{t-1,2}r_{t-1}(a^\beta-(1-q_{t-1,1})(a^\beta-q_{t-2,2}c^\beta)))
	\end{aligned}
	\label{equation:equation2_for_t_less_n}
	\end{equation}
	
	If $\gamma[-1]=1$ and $\beta[-1]=1$, we have
	\begin{equation}
	\begin{aligned}
	a^{\overline{\gamma 0}}&=(1-r_tq_{t,1})(1-p_{t-1})(a^\gamma-q_{t-1,2}c^\gamma)\\
	&=(1-r_tq_{t,1})(1-p_{t-1})\\&
	(a^\gamma-q_{t-1,2}r_{t-1}(a^\beta-(1-p_{t-2})(1-q_{t-1,1})(a^\beta-q_{t-2,2}c^\beta)))
	\end{aligned}
	\end{equation}
	
	Each of these equations is only one more factor compared to the $t=n$ case, so we can still get the similar result
	\begin{equation}
	\begin{aligned}
	&(\frac{\text{Co}_1(\beta_1)}{a^{\overline{\beta_110}}}-\frac{\text{Co}_1(\beta_2)}{a^{\overline{\beta_210}}})q_{t-1,2}r_{t-1}
	+(\frac{\text{Co}_2(\beta_1)}{a^{\overline{\beta_110}}}-\frac{\text{Co}_2(\beta_2)}{a^{\overline{\beta_210}}})q_{t-1,2}\\
	&-(\frac{\text{Co}_1(\beta_1)+\text{Co}_2(\beta_1)}{a^{\overline{\beta_110}}}-\frac{\text{Co}_1(\beta_2)+\text{Co}_2(\beta_2)}{a^{\overline{\beta_210}}})=0
	\end{aligned}
	\end{equation}
	which is an equation of $q_{t-1,2}r_{t-1}$ and $q_{t-1,2}$. Since we assume that $r_{t-1}$ or $q_{t-1,2}$ is known, we can get both $r_{t-1}$ and $q_{t-1,2}$ by this equation. Then we can get $r_{t}q_{t,1}$ by Equation \ref{equation:equation1_for_t_less_n} using $r_tq_{t,1}=1-\frac{a^{\overline{\gamma 0}}}{a^\gamma-\frac{q_{t-1,2}r_{t-1}(1-q_{t-1,1})}{1-q_{t-1,1}r_{t-1}}a^\gamma}$. Then all the terms in Equation \ref{equation:equation2_for_t_less_n} except $p_{t-1}$ are known and not $0$, so we can get $p_{t-1}$. This completes the inductive transition.
	
	Finally, we consider the base case. We only need to prove that if $p_1$ is known, one of $r_1$ and $q_{1,2}$ is known, then $q_{1,1},r_2q_{2,1}$ and the unknown one between $r_1,q_{1,2}$ can be computed using $a^{\overline{00}},a^{\overline{01}}$ and $a^{\overline{10}}$. Actually, we have the following equations
	
	\begin{align}
	&a^{\overline{00}}=(1-r_1)(1-r_2q_{2,1})+r_1(1-q_{1,1})(1-q_{1,2})(1-r_2q_{2,1}),\\
	&a^{\overline{10}}=r_1q_{1,1}(1-q_{1,2})(1-r_2q_{1,2})(1-p_1),\\
	&a^{\overline{01}}=(1-r_1)r_2q_{2,1}+r_1(1-q_{1,1})(1-(1-q_{1,2})(1-r_2q_{2,1})).
	\end{align}
	
	When we know $p_1$ and $r_1$, we can get the other parameters as below
	\begin{small}
		\begin{align}
		&q_{1,2}=\frac{1-a^{\overline{00}}-a^{\overline{01}}}{r_1},\\
		&r_2q_{2,1}=((a^{\overline{00}})^2(p_1-1)+(1-r_1)(p_1+a^{\overline{10}}-1)-a^{\overline{01}}(a^{\overline{01}}-1+p_1+r_1-p_1r_1)\\&-a^{\overline{00}}(-2+a^{\overline{01}}+a^{\overline{10}}+2p_1-a^{\overline{01}}p_1+r_1-p_1r_1))\\&/((a^{\overline{00}}+a^{\overline{01}}-1)(p_1-1)(r_1-1)r_2),\\
		&q_{1,1}=1+\frac{a^{\overline{10}}r_1}{(-1+p_1)q_{1,1}(-1+r_1)r_1(q_{2,1}r_2-1)}.
		\end{align}
	\end{small}
	
	When we know $p_1$ and $q_{1,2}$, we can get the other parameters as below
	\begin{small}
		\begin{align}
		&r_1=\frac{1}{a^{\overline{10}}q_{1,2}}(-a^{\overline{00}}+(a^{\overline{00}})^2+a^{\overline{00}}a^{\overline{01}}+a^{\overline{00}}a^{\overline{10}}+a^{\overline{01}}a^{\overline{10}}+a^{\overline{00}}p_1-(a^{\overline{00}})^2p_1\\&-a^{\overline{00}}a^{\overline{01}}p_1+a^{\overline{00}}q_{1,2}-(a^{\overline{00}})^2q_{1,2}-a^{\overline{00}}a^{\overline{01}}q_{1,2}+a^{\overline{10}}q_{1,2}-a^{\overline{00}}a^{\overline{10}}q_{1,2}\\&-a^{\overline{01}}a^{\overline{10}}q_{1,2}+(a^{\overline{00}})^2p_1q_{1,2}+a^{\overline{00}}a^{\overline{01}}p_1q_{1,2}),\\
		&q_{1,1}=\frac{1-a^{\overline{00}}-a^{\overline{01}}}{r_1},\\
		&r_2q_{2,1}=\frac{-1 + a^{\overline{10}} + p_1 + (a^{\overline{00}}+a^{\overline{01}}) (-1 + p_1) (-1 + q_{1,2}) + q_{1,2} - p_1 q_{1,2}}{(-1+a^{\overline{00}}+a^{\overline{01}})(-1+p_1)(-1+q_{1,2})r_2}.
		\end{align}
	\end{small}
	Notice that all the denominators in these solutions are not $0$, therefore, we have proved the base case.
	Until now, we have proved the whole theorem.\qed
\end{proof}

\section{Proof for the Model with a Global Hidden Variable  (Theorem~\ref{thm:id_of_glob})}
\label{app:thmidofglob}
\thmidofglob*

\begin{proof}
We have the following equations
\begin{small}
\begin{equation}
    \begin{aligned}
    &P(V_1=0,V_2=0,\cdots,V_t=0)=(1-q_1)(1-q_2)\cdots (1-q_t)r+(1-r),\\
    &P(V_1=0,V_2=0,\cdots,V_t=1,V_{t+1}=0,\cdots,V_{t_0-1}=0)\\&=(1-q_1)(1-q_2)\cdots q_t(1-q_{t+1})\cdots(1-q_{t_0-1})r\prod_{(V_t,V_i)\in E,t+1\leq i<t_0}(1-p_{t,i}),\\
    &P(V_1=0,V_2=0,\cdots,V_t=1,V_{t+1}=0,\cdots,V_{t_0}=0)\\&=(1-q_1)(1-q_2)\cdots q_t(1-q_{t+1})\cdots (1-q_{t_0}))r\prod_{(V_t,V_i)\in E,t+1\leq i<t_0}(1-p_{t,i})
    \end{aligned}
\end{equation}\end{small}
such that $V_{t_0}$ is not a child of $V_{t}$. Therefore, for each pair of unconnected nodes, if two of them $V_i,V_j$ ($i<j$) satisfy $q_i,q_j\neq 0$, we can get $1-q_j=\frac{P(V_1=0,V_2=0,\cdots,V_i=1,V_{i+1}=0,\cdots,V_j=0)}{P(V_1=0,V_2=0,\cdots,V_i=1,V_{i+1}=0,\cdots,V_{j-1}=0)}$. Then we can use the following two equations to solve $r$:
\begin{equation}
    \begin{aligned}
    &(1-q_j)(r(1-q_1)(1-q_2)\cdots (1-q_{j-1}))+(1-r)\\&=P(V_1=0,V_2=0,\cdots,V_j=0),\\&(r(1-q_1)(1-q_2)\cdots (1-q_{j-1}))+(1-r)\\&=P(V_1=0,V_2=0,\cdots,V_{j-1}=0).
    \end{aligned}
\end{equation}
If we see $1-r$ and $r(1-q_1)(1-q_2)\cdots (1-q_{j-1})$ as two unknown variables, they can be solved by this system of linear equations. With the value of $r$, we can solve out $q_1,q_2,\cdots,q_n$ by using $P(V_1=0,V_1=0,\cdots,V_t=1)=(1-q_1)(1-q_2)\cdots q_tr$ for $t=1,2,\cdots,n$. For $p_{i,j}$, it can be computed by $1-p_{i,j}=\frac{P(V_1=0,V_2=0,\cdots,V_i=1,V_{i+1}=0,\cdots,V_j=0)}{P(V_1=0,V_2=0,\cdots,V_i=1,V_{i+1}=0,\cdots,V_{j}=0)(1-q_{j})}$ if $q_i\neq 0$. We suppose the parents of $V_j$ are $V_{i_1},V_{i_2},\cdots,V_{i_t}$ such that $i_1<i_2<\cdots <i_t<j$. Therefore, we have
\begin{equation}
    \begin{aligned}
    &\frac{P(V_{i_1}=0,\cdots,V_{i_{k-1}}=0,V_{i_k}=1,V_{i_{k+1}},\cdots,V_{i_t}=0,V_j=0)}{P(V_{i_1}=0,\cdots,V_{i_{k-1}}=0,V_{i_k}=1,V_{i_{k+1}},\cdots,V_{i_t}=0}\\&=(1-q_j)(1-p_{i_k,j}),k=1,2,\cdots,t
    \end{aligned}
\end{equation}
if $P(V_{i_k}=1)>0$. Hence, we can still get all the meaningful $p_{i_k,j}$. Until now, we have got all the parameters solved.

Now we consider the case that $\not\exists 1\leq i<j\leq n$ such that $(V_i,V_j)\not\in E$ and $q_i,q_j\neq 0$. This is equivalent to the claim that for all the $V_i$ such that $q_i\neq 0$, they form a complete graph if we remove all the directions of edges. Suppose three of them are $V_i,V_j,V_k,1\leq i<j<k\leq n$. Then we can get the following equations:
\begin{scriptsize}
\begin{equation}
    \begin{aligned}
    &\frac{P(V_1=0,\cdots,V_{i-1}=0,V_i=1,V_{i+1}=0,\cdots,V_j=0)}{P(V_1=0,\cdots,V_{i-1}=0,V_i=1,V_{i+1}=0,\cdots,V_{j-1}=0)}=(1-q_j)(1-p_{i,j}),\\
    &\frac{P(V_1=0,\cdots,V_{i-1}=0,V_i=1,V_{i+1}=0,\cdots,V_k=0)}{P(V_1=0,\cdots,V_{i-1}=0,V_i=1,V_{i+1}=0,\cdots,V_{k-1}=0)}=(1-q_k)(1-p_{i,k}),\\
    &\frac{P(V_1=0,\cdots,V_{j-1}=0,V_j=1,V_{j+1}=0,\cdots,V_k=0)}{P(V_1=0,\cdots,V_{j-1}=0,V_j=1,V_{j+1}=0,\cdots,V_{k-1}=0)}=(1-q_k)(1-p_{j,k}),\\
    &\frac{P(V_1=0,\cdots,V_{i-1}=0,V_i=1,V_{i+1}=0,\cdots,V_{j-1}=0,V_j=1,V_{j+1}=0,\cdots,V_k=0)}{P(V_1=0,\cdots,V_{i-1}=0,V_i=1,V_{i+1}=0,\cdots,V_{k-1}=0)}\\&=\frac{(1-q_k)(1-p_{i,k})(1-p_{j,k})(1-(1-q_j)(1-p_{i,j}))}{(1-q_j)(1-p_{i,j})}.
    \end{aligned}
\end{equation}
\end{scriptsize}
Therefore, $q_i,q_j,q_k,p_{i,j},p_{i,k},p_{j,k}$ can all be solved. Then we can use the same procedure to get all the $q_1,q_2,\cdots,q_n$ and then all the parameters.

In conclusion, we have solved the identifiability problem of the model with a global hidden variable.\qed
\end{proof}

\section{Proof for the Mixed Model (Theorem~\ref{thm:mixedmodel})}

\thmmixedmodel*

\begin{proof}
	In order to simplify our proof, we introduce some new notations. Suppose we already have $V_i,V_j,V_k$ as required in description of this theorem.
	\begin{equation}
	\begin{aligned}
	&a_l=r\prod_{t=1}^l(1-q_t)(1-q_{0,t})=P(V_1=0,V_2=0,\cdots,V_l=0,U_0=1),\\
	&b_l=(1-r)\prod_{t=1}^l(1-q_t)=P(V_1=0,V_2=0,\cdots,V_l=0,U_0=0),
	\end{aligned}
	\end{equation}
	\begin{equation}
	\begin{aligned}
	&x_{i,l}=\frac{1-(1-q_i)(1-q_{0,i})}{(1-q_i)(1-q_{0,i})}\prod_{(V_i,V_t)\in E,t\leq l}(1-p_{i,t})\\&=\frac{P(V_1=0,V_2=0,\cdots,V_l=0,U_0=1)}{P(V_1=0,\cdots,V_{i-1}=0,V_i=1,V_{i+1}=0,\cdots,V_l=0,U_0=1)},\\
	&x_{j,l}=\frac{1-(1-q_j)(1-q_{0,j})}{(1-q_j)(1-q_{0,j})}\prod_{(V_j,V_t)\in E,t\leq l}(1-p_{j,t})\\&=\frac{P(V_1=0,V_2=0,\cdots,V_l=0,U_0=1)}{P(V_1=0,\cdots,V_{j-1}=0,V_j=1,V_{j+1}=0,\cdots,V_l=0,U_0=1)},\\
	&x_{k,l}=\frac{1-(1-q_k)(1-q_{0,k})}{(1-q_k)(1-q_{0,k})}\prod_{(V_k,V_t)\in E,t\leq l}(1-p_{k,t})\\&=\frac{P(V_1=0,V_2=0,\cdots,V_l=0,U_0=1)}{P(V_1=0,\cdots,V_{k-1}=0,V_k=1,V_{k+1}=0,\cdots,V_l=0,U_0=1)},
	\end{aligned}
	\end{equation}
	\begin{equation}
	\begin{aligned}
	&y_{i,l}=\frac{q_i}{1-q_i}\prod_{(V_i,V_t)\in E,t\leq l}(1-p_{i,t})\\&=\frac{P(V_1=0,V_2=0,\cdots,V_l=0,U_0=1)}{P(V_1=0,\cdots,V_{i-1}=0,V_i=1,V_{i+1}=0,\cdots,V_l=0,U_0=0)},\\
	&y_{j,l}=\frac{q_j}{1-q_j}\prod_{(V_j,V_t)\in E,t\leq l}(1-p_{j,t})\\&=\frac{P(V_1=0,V_2=0,\cdots,V_l=0,U_0=1)}{P(V_1=0,\cdots,V_{j-1}=0,V_j=1,V_{j+1}=0,\cdots,V_l=0,U_0=0)},\\
	&y_{k,l}=\frac{q_k}{1-q_k}\prod_{(V_k,V_t)\in E,t\leq l}(1-p_{k,t})\\&=\frac{P(V_1=0,V_2=0,\cdots,V_l=0,U_0=1)}{P(V_1=0,\cdots,V_{k-1}=0,V_k=1,V_{k+1}=0,\cdots,V_l=0,U_0=0)}.
	\end{aligned}
	\end{equation}
	Therefore, we can easily verify that $a_l+b_l=P(V_1=0,V_2=0,\cdots,V_l=0)\equiv p_1$. Moreover, because $V_i,V_j,V_k$ are disconnected with each other, we have the following equations.
	\begin{small}
		\begin{equation}
		\begin{aligned}
		&a_lx_{i,l}+b_ly_{i,l}=P(V_1=0,\cdots,V_{i-1}=0,V_i=1,V_{i+1}=0,\cdots,V_l=0)\equiv p_2,\\
		&a_lx_{j,l}+b_ly_{j,l}=P(V_1=0,\cdots,V_{j-1}=0,V_j=1,V_{j+1}=0,\cdots,V_l=0)\equiv p_3,\\
		&a_lx_{k,l}+b_ly_{k,l}=P(V_1=0,\cdots,V_{k-1}=0,V_k=1,V_{k+1}=0,\cdots,V_l=0)\equiv p_4,\\
		&a_lx_{i,l}x_{j,l}+b_ly_{i,l}y_{j,l}=P(\overline{V_1V_2\cdots V_l}=\overline{0^{i-1}10^{j-i-1}10^{l-j}})\equiv p_5,\\
		&a_lx_{i,l}x_{k,l}+b_ly_{i,l}y_{k,l}=P(\overline{V_1V_2\cdots V_l}=\overline{0^{i-1}10^{k-i-1}10^{l-k}})\equiv p_6,\\
		&a_lx_{k,l}x_{j,l}+b_ly_{k,l}y_{j,l}=P(\overline{V_1V_2\cdots V_l}=\overline{0^{j-1}10^{k-j-1}10^{l-k}})\equiv p_7,\\
		&a_lx_{i,l}x_{j,l}x_{k,l}+b_ly_{i,l}y_{j,l}y_{k,l}=P(\overline{V_1V_2\cdots V_l}=\overline{0^{i-1}10^{j-i-1}10^{k-j-1}10^{l-k}})\equiv p_8.
		\end{aligned}
		\label{equation:mixed_equations}
		\end{equation}
	\end{small}
	Here, $p_t,t\in[8]$ are known because $V_1,V_2,\cdots,V_l$ are observable. Therefore, we can solve out $a_l,b_l$ using the $8$ equations. The result is shown as below.
	\begin{equation}
	\begin{aligned}
	&a_l=\frac{\text{numerator}_1\pm\text{numerator}_2}{\text{denominator}},\\
	&\text{denominator}=2(-p_1^2 p_8^2+2 p_1 p_2 p_7 p_8+2
	p_1 p_3 p_6 p_8+2 p_1 p_4 p_5 p_8\\&-4
	p_1 p_5 p_6 p_7-p_2^2 p_7^2-4 p_2 p_3
	p_4 p_8+2 p_2 p_3 p_6 p_7+2 p_2 p_4
	p_5 p_7-p_3^2 p_6^2\\&+2 p_3 p_4 p_5
	p_6-p_4^2 p_5^2),\\
	&\text{numerator}_1=2 {p_1}^2 {p_2} {p_7} {p_8}+2
	{p_1}^2 {p_3} {p_6} {p_8}+2 {p_1}^2 {p_4} {p_5} {p_8}-4
	{p_1}^2 {p_5} {p_6} {p_7}-\\&{p_1} {p_2}^2 {p_7}^2-4 {p_1}
	{p_2} {p_3} {p_4} {p_8}+2 {p_1} {p_2} {p_3} {p_6}
	{p_7}+2 {p_1} {p_2} {p_4} {p_5} {p_7}-{p_1} {p_3}^2
	{p_6}^2\\&+2 {p_1} {p_3} {p_4} {p_5} {p_6}-{p_1} {p_4}^2
	p_5^2-p_1^3p_8^2,\\
	&\text{numerator}_2=\left({p_1}^2 {p_8}-{p_1} {p_2}
	{p_7}-{p_1} {p_3} {p_6}-{p_1} {p_4} {p_5}+2 {p_2}
	{p_3} {p_4}\right)\\& ({p_1}^2 {p_8}^2-2 {p_1} {p_2} {p_7}
	{p_8}-2 {p_1} {p_3} {p_6} {p_8}-2 {p_1} {p_4} {p_5}
	{p_8}+4 {p_1} {p_5} {p_6} {p_7}+{p_2}^2 {p_7}^2\\&+4 {p_2}
	{p_3} {p_4} {p_8}-2 {p_2} {p_3} {p_6} {p_7}-2 {p_2}
	{p_4} {p_5} {p_7}+{p_3}^2 {p_6}^2-2 {p_3} {p_4} {p_5}
	{p_6}+{p_4}^2 {p_5}^2)^\frac{1}{2}.
	\end{aligned}
	\end{equation}
	We can verify that
	\begin{align}
	&(p_2 p_3 - p_1 p_5) (p_2 p_4 - p_1 p_6) (p_3 p_4 - p_1 p_7)\\&=-a_l^3b_l^3(x_{i,l}-y_{i,l})^2(x_{j,l}-y_{j,l})^2(x_{k,l}-y_{k,l})^2<0
	\end{align}
	Therefore, there is only one solution of $a_l$ which is positive because 
	\begin{align}
	&\text{numerator}_1^2-\text{numerator}_2^2\\
	&=-4(p_2 p_3 - p_1 p_5) (p_2 p_4 - p_1 p_6) (p_3 p_4 - p_1 p_7)\\
	&(p_4^2 p_5^2 - 2 p_3 p_4 p_5 p_6 + p_3^2 p_6^2 - 2 p_2 p_4 p_5 p_7 - 
	2 p_2 p_3 p_6 p_7 + 4 p_1 p_5 p_6 p_7 + p_2^2 p_7^2 \\&+ 4 p_2 p_3 p_4 p_8 - 
	2 p_1 p_4 p_5 p_8 - 2 p_1 p_3 p_6 p_8 - 2 p_1 p_2 p_7 p_8 + p_1^2 p_8^2)
	\\&=2(p_2 p_3 - p_1 p_5) (p_2 p_4 - p_1 p_6) (p_3 p_4 - p_1 p_7)\text{denominator},\\
	&\text{numerator}_1=\frac{1}{2}p_1\text{denominator}.
	\end{align}
	Because $a_l$ is fixed now, we can get $b_l$ according to $a_l+b_l=p_1$. This process can be done if and only if the denominator is not zero, which is equivalent to $a^2 b^2 (x_i - y_i)^2 (x_j - y_j)^2 (x_k - y_k)^2\neq 0$. Therefore, we only need to find $V_i,V_j,V_k$ satisfying the conditions in the theorem and $q_{0,i},q_{0,j},q_{0,k}\neq 0$.
	
	Notice that $l$ is an arbitrary number not less than $k$, we can get $q_t$ and $q_{0,t},k+1\leq n$ can be computed using $\frac{a_{t}}{a_{t-1}}=(1-q_t)(1-q_{0,t})$ and $\frac{b_t}{b_{t-1}}=1-q_t$. For parameter $p_{t,l},l>t>k$, we can compute it using 
	\begin{align}
	&\frac{P(V_1=0,\cdots,V_{t-1}=0,V_t=1,V_{t+1}=0,\cdots,V_l=0)}{P(V_1=0,\cdots,V_l=0)}\\&=\frac{a_l\frac{1-(1-q_t)(1-q_{0,t})}{(1-q_t)(1-q_{0,t})}+b_l\frac{q_t}{1-q_t}}{a_l+b_l}\prod_{(V_t,V_i)\in E,i\leq l}(1-p_{t,i}),\\
	&\frac{P(V_1=0,\cdots,V_{t-1}=0,V_t=1,V_{t+1}=0,\cdots,V_{l-1}=0)}{P(V_1=0,\cdots,V_{l-1}=0)}\\&=\frac{a_{l-1}\frac{1-(1-q_t)(1-q_{0,t})}{(1-q_t)(1-q_{0,t})}+b_{l-1}\frac{q_t}{1-q_t}}{a_{l-1}+b_{l-1}}\prod_{(V_t,V_i)\in E,i\leq l-1}(1-p_{t,i}).
	\end{align}
	Then we can get $\prod_{(V_t,V_i)\in E,i\leq l-1}(1-p_{t,i})$ and $\prod_{(V_t,V_i)\in E,i\leq l}(1-p_{t,i})$ and their division is what we want. Until now, we have proved the first part of the theorem. Now suppose we have $i=j-1=k-2$. Then for an arbitrary $0-1$ string $\gamma$ with $i-1$ bits, we still have the following facts.
	\begin{tiny}
		\begin{equation}
		\begin{aligned}
		&\frac{P(\overline{V_1\cdots V_{i-1}}=\gamma,V_i=1,V_{j}=1,V_{k}=1,U_0=t)}{P(\overline{V_1\cdots V_{i-1}}=\gamma,V_i=0,V_{j}=0,V_{k}=0,U_0=t)}=\frac{P(\overline{V_1\cdots V_{i-1}}=\gamma,V_i=1,V_{j}=0,V_{k}=0,U_0=t)}{P(\overline{V_1\cdots V_{i-1}}=\gamma,V_i=0,V_{j}=0,V_{k}=0,U_0=t)}\\&\frac{P(\overline{V_1\cdots V_{i-1}}=\gamma,V_i=0,V_{j}=1,V_{k}=0,U_0=t)}{P(\overline{V_1\cdots V_{i-1}}=\gamma,V_i=0,V_{j}=0,V_{k}=0,U_0=t)}\frac{P(\overline{V_1\cdots V_{i-1}}=\gamma,V_i=0,V_{j}=0,V_{k}=1,U_0=t)}{P(\overline{V_1\cdots V_{i-1}}=\gamma,V_i=0,V_{j}=0,V_{k}=0,U_0=t)},\\
		&\frac{P(\overline{V_1\cdots V_{i-1}}=\gamma,V_i=1,V_{j}=1,V_{k}=0,U_0=t)}{P(\overline{V_1\cdots V_{i-1}}=\gamma,V_i=0,V_{j}=0,V_{k}=0,U_0=t)}=\frac{P(\overline{V_1\cdots V_{i-1}}=\gamma,V_i=1,V_{j}=0,V_{k}=0,U_0=t)}{P(\overline{V_1\cdots V_{i-1}}=\gamma,V_i=0,V_{j}=0,V_{k}=0,U_0=t)}\\&\frac{P(\overline{V_1\cdots V_{i-1}}=\gamma,V_i=0,V_{j}=1,V_{k}=0,U_0=t)}{P(\overline{V_1\cdots V_{i-1}}=\gamma,V_i=0,V_{j}=0,V_{k}=0,U_0=t)},\\
		&\frac{P(\overline{V_1\cdots V_{i-1}}=\gamma,V_i=1,V_{j}=0,V_{k}=1,U_0=t)}{P(\overline{V_1\cdots V_{i-1}}=\gamma,V_i=0,V_{j}=0,V_{k}=0,U_0=t)}=\frac{P(\overline{V_1\cdots V_{i-1}}=\gamma,V_i=1,V_{j}=0,V_{k}=0,U_0=t)}{P(\overline{V_1\cdots V_{i-1}}=\gamma,V_i=0,V_{j}=0,V_{k}=0,U_0=t)}\\&\frac{P(\overline{V_1\cdots V_{i-1}}=\gamma,V_i=0,V_{j}=0,V_{k}=1,U_0=t)}{P(\overline{V_1\cdots V_{i-1}}=\gamma,V_i=0,V_{j}=0,V_{k}=0,U_0=t)},\\
		&\frac{P(\overline{V_1\cdots V_{i-1}}=\gamma,V_i=0,V_{j}=1,V_{k}=1,U_0=t)}{P(\overline{V_1\cdots V_{i-1}}=\gamma,V_i=0,V_{j}=0,V_{k}=0,U_0=t)}=\frac{P(\overline{V_1\cdots V_{i-1}}=\gamma,V_i=0,V_{j}=0,V_{k}=1,U_0=t)}{P(\overline{V_1\cdots V_{i-1}}=\gamma,V_i=0,V_{j}=0,V_{k}=0,U_0=t)}\\&\frac{P(\overline{V_1\cdots V_{i-1}}=\gamma,V_i=0,V_{j}=1,V_{k}=0,U_0=t)}{P(\overline{V_1\cdots V_{i-1}}=\gamma,V_i=0,V_{j}=0,V_{k}=0,U_0=t)}
		\end{aligned}
		\end{equation}
	\end{tiny}
	where $t=0$ or $t=1$. Therefore, we can still have that $8$ equations the same as Equation \ref{equation:mixed_equations} and from those equations, we can solve out all of the probabilities in the facts above. Using those probabilities, we can directly get all the parameters with indexes not larger than $k$. Together with the conclusion of the first part of this theorem, all the parameters are determined.\qed
\end{proof}

\section{Discussion on the Cyclic Models}
\label{app:cyclic}

In the last section in appendix, we will introduce how to transform a general IC model to a causal model so that we can use do calculus method to identify do effects. In fact, since the propagation of the IC model takes place for at most $n$ rounds \cite{chen2013information}, we formulate that the state of $V_i$ in round $t$ is $V_{i,t}$ and that $V_{i,t}$ has three values, $0,1$ and $2$, for three states. We construct a causal graph $G'=(V',E')$ from these nodes with subscripts of time. Here, state $0$ means that the node is not activated, $1$ means that the node was activated at the last time point, and state $2$ means that the node is activated and has already tried to activate its child nodes. Moreover, all edges are defined as $(V_{i,t},V_{i,t+1}),1\leq i,t\leq n$ and $(V_{i,t},V_{j,t+1}),1\leq i,j,t\leq n,(V_i,V_j)\in E'$. Furthermore, we define the following propagating equations.
\begin{align}
    &P(V_{i,t}=2|V_{i,t-1}=2,V_{j,t-1}=v_{j,t-1},\forall j, (V_j,V_i)\in E)=1,\\
    &P(V_{i,t}=2|V_{i,t-1}=1,V_{j,t-1}=v_{j,t-1},\forall j, (V_j,V_i)\in E)=1,\\
    &P(V_{i,t}=2|V_{i,t-1}=0,V_{j,t-1}=v_{j,t-1},\forall j, (V_j,V_i)\in E)=0,\\
    &P(V_{i,t}=1|V_{i,t-1}=0,V_{j,t-1}=v_{j,t-1},\forall j, (V_j,V_i)\in E)\\&=1-\prod_{1\leq j\leq n,(V_j,V_i)\in E}(1-p_{j,i}v_{j,t-1}I[v_{j,t-1}\neq 2]),\\
    &P(V_{i,t}=0|V_{i,t-1}=0,V_{j,t-1}=v_{j,t-1},\forall j, (V_j,V_i)\in E)\\&=\prod_{1\leq j\leq n,(V_j,V_i)\in E}(1-p_{j,i}v_{j,t-1}I[v_{j,t-1}\neq 2]).
\end{align}

By definition, we obtain that $G'$ as a Bayesian causal model, so we can use the do calculus method \cite{Pearl09} to solve out the identifiable do effects. For example, suppose $G$ has $V_1,V_2,V_3$ as observed nodes and $U_1$ as an unobserved node. The edges in $E$ are $(U_1,V_1),(U_1,V_2),(V_1,V_2),(V_2,V_3)$ and $(V_3,V_1)$. There exists a cycle in $G$ so it cannot be seen as a DAG. However, utilizing our transformation, the result graph $G'$ is in Figure \ref{fig:general_model}, which is a Bayesian causal model.

\begin{figure}[tb]
    \centering
    \includegraphics[width=0.35\linewidth]{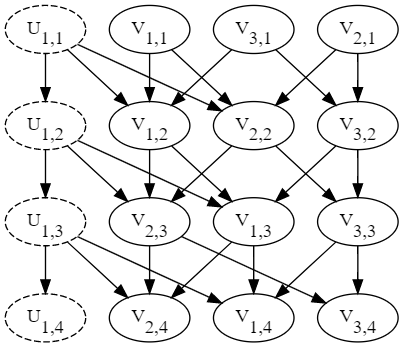}
    \caption{An example of transformation from IC model to Bayesian causal graph.}
    \label{fig:general_model}
\end{figure}
}

\end{document}